\documentclass[conference,letterpaper]{IEEEtran}
\IEEEoverridecommandlockouts

\usepackage{graphicx}
\usepackage{amsmath,amsthm,amsfonts,amssymb}
\usepackage{algorithmic}    
\usepackage{cite}
\usepackage{bm}
\usepackage{bbm}
\usepackage{url}
\usepackage{array}
\usepackage{color,soul}
\usepackage{multirow}
\usepackage{booktabs}
\usepackage[table,xcdraw]{xcolor}
\usepackage{enumitem}
\usepackage[caption=false,font=footnotesize]{subfig}

\theoremstyle{plain}
\newtheorem{thm}{Theorem}

\newtheorem{prop}[thm]{Proposition}

\newtheorem{remark}{Remark}

\newcommand{\supp}{\operatorname{supp}}

\newcommand{\Top}{\operatorname{Top}}
\newcommand{\argmin}{\operatornamewithlimits{arg\,min}}

\usepackage[english]{babel}

\usepackage{algorithm}

\usepackage{comment}         
\allowdisplaybreaks[4]
\usepackage{cases}   
\begin{document}
\title{Delay-Doppler Domain Channel Estimation: \\ What if Sparsity is Unknown?}
\author{Zijian Yang, Yulin Shao, Fen Hou, Shaodan Ma
\thanks{Z. Yang, F. Hou, and S. Ma are with the State Key Laboratory of Internet of Things for Smart City, University of Macau, Macau, China (e-mails: \{yc37408,shaodanma,fenhou\}@um.edu.mo).}
\thanks{Y. Shao is with the Department of Electrical and Computer Engineering, The University of Hong Kong, Hong Kong, China (e-mail: ylshao@hku.hk).}
}
\maketitle

\begin{abstract}
Sparsity in the delay-Doppler (DD) domain enables efficient channel estimation, but the realization-wise sparsity level is rarely known in advance, and it fluctuates. What if we could estimate the channel without ever knowing how many delays or Dopplers are active? This paper answers that question. We propose a sparsity-agnostic structured estimator that requires no prior knowledge of delay or Doppler sparsity budgets. The key idea is to exploit the Cartesian-product structure of DD support (active delays share a common Doppler set) and to select the support dimensions directly from the data via the Bayesian information criterion. We instantiate the framework on an affine frequency division multiplexing system, where the observation model naturally admits an on-grid DD representation.  Numerical results demonstrate that it recovers the exact support with high probability and achieves near-oracle channel reconstruction accuracy, consistently outperforming fixed-budget baselines and sparse Bayesian learning. The approach is waveform-agnostic and offers a practical, adaptive solution for DD-domain channel estimation under unknown and time-varying sparsity.
\end{abstract}


\section{Introduction}
\label{sec:introduction}

Sparsity is a fundamental characteristic of wireless channels, especially in high-mobility and doubly dispersive propagation environments where only a limited number of dominant propagation components carry most of the channel energy \cite{bajwa2010compressed,molisch2004geometry,zhang2025polarization,cao2026stepped}. When such channels are represented in the delay-Doppler (DD) domain, their energy often concentrates on a small subset of DD coefficients, making sparse DD-domain modeling a natural foundation for pilot-efficient channel estimation \cite{taubock2010compressive}. Under an on-grid approximation, the estimation problem reduces to a sparse linear inverse problem, enabling the application of compressed sensing and sparse recovery techniques \cite{berger2010application,eldar2012compressed}.

Despite the maturity of sparse recovery theory, a critical practical issue remains: the realization-wise sparsity level of the DD-domain channel is generally unknown and can vary significantly from one channel realization to another. This is not merely a technical nuance: in realistic doubly dispersive channels, the number of active delay taps and the number of active Doppler components fluctuate over time and across different propagation environments. As a consequence, many existing estimators still rely, either explicitly or implicitly, on externally prescribed sparsity budgets, manually tuned thresholds, or fixed model orders \cite{tropp2007signal,benzine2026models,shao2026embodied}. Unfortunately, such quantities are rarely known a priori in practice, and mismatches between the assumed and true sparsity levels lead to either support underestimation (missing real paths) or overestimation (introducing spurious coefficients).

A parallel difficulty is that generic sparse recovery algorithms treat the DD-domain channel as an unstructured sparse vector, overlooking the inherent regularity of the DD support pattern \cite{eldar2012compressed}. In practice, however, the support is not arbitrary: active delay taps often share a common set of Doppler frequencies, leading to a Cartesian-product structure \cite{duarte2011structured,3gpp}. Such regularity, when properly exploited, can significantly reduce the ambiguity of support search and improve model selection. Therefore, an important open question is: \textit{Can we design a DD-domain channel estimator that (i) requires no prior knowledge of sparsity levels, (ii) exploits the structured nature of DD support, and (iii) adapts to each channel realization in a principled manner?}

In this paper, we answer this question affirmatively by developing a sparsity-agnostic structured channel estimator for the DD domain. Instead of assuming known delay and Doppler sparsity levels, we infer the effective support dimensions directly from the received observations. Rather than applying a generic sparse recovery routine over an unrestricted support space, we exploit the Cartesian-product regularity of the DD support and construct a structure-aware search strategy. The final model is selected via the Bayesian information criterion (BIC) \cite{schwarz1978estimating}, which balances data fitting accuracy and model complexity without requiring any external sparsity budget.

The proposed method is deliberately decoupled from any specific waveform. It requires only a linear observation model relating the received observations to the vectorized DD-domain channel via a known sensing matrix, under additive noise. Such a model arises naturally in various multicarrier systems, including but not limited to AFDM, OTFS, and OFDM with appropriate preprocessing. In this paper, we instantiate the framework using AFDM \cite{bemani2023affine,cao2025agile}, because AFDM provides a particularly clean DD-domain input-output relation, but the core methodology is transferable to any system that admits an on-grid DD observation model.

Our main contributions are summarized as follows.
\begin{itemize}[leftmargin=0.5cm]
    \item We propose a sparsity-agnostic DD-domain channel estimator that does not require the delay and Doppler sparsity levels to be known a priori. This makes the method practical for realistic channels where the realization-wise sparsity fluctuates from block to block.
    \item By exploiting the Cartesian-product regularity of DD-domain support, we develop a structure-aware support search strategy that significantly reduces the search space compared to unstructured sparse recovery, while remaining robust to unknown model order.
    \item We establish a BIC-based model selection mechanism that automatically determines the support dimensions and the final support candidate from the data. Numerical results demonstrate that the proposed method achieves high exact support recovery probability and near-oracle channel reconstruction accuracy, outperforming fixed-budget baselines and sparse Bayesian learning.
\end{itemize}

\section{System Model}
\label{sec:system_model}

\subsection{Doubly Selective Channel Estimation}

We first describe the discrete-time doubly selective channel model and its on-grid DD representation.
Let $s_n$ and $r_n$ denote the transmitted and received discrete-time signals, respectively. The doubly selective channel is modeled as
\begin{equation}
r_n = \sum_{l=0}^{L-1} s_{n-l} h_{l,n} + z_n,\qquad n\in\mathbb{Z},
\label{eq:sys_rn}
\end{equation}
where $L-1$ is the maximum channel delay, $h_{l,n}$ is the time-varying complex gain associated with the $\ell$-th delay tap at time index $n$ \cite{shao2021federated}, and $z_n \sim \mathcal{CN}(0,\sigma_w^2)$ is additive white Gaussian noise (AWGN). 
The goal of channel estimation is to recover $\{h_{l,n}\}$ from the received signal.

To make the problem tractable, we adopt an on-grid DD representation. Each delay tap is expanded over a finite Doppler grid:
\begin{equation}
h_{l,n}
=
\sum_{q=-Q}^{Q}
\alpha_{l,q} e^{j\frac{2\pi}{N}nq},
\quad l=0,\ldots,L-1,
\end{equation}
where $q$ is the Doppler-bin index, $Q$ is the maximum normalized Doppler index, and $\alpha_{l,q}$ denotes the effective DD coefficient at location $(l,q)$. 
The coefficient is further modeled as $\alpha_{l,q}=I_{l,q}g_{l,q}$, with $I_{l,q}\in\{0,1\}$ indicating activity and $g_{l,q}\sim\mathcal{CN}(0,\sigma_\alpha^2)$ the complex gain.

For convenience, we define $B\triangleq 2Q+1$ such that each delay tap corresponds to a Doppler block of size $B$. The $\ell$-th block coefficient vector is defined as
$\bm{\alpha}_l\triangleq[\alpha_{l,-Q},\alpha_{l,-Q+1},\ldots,
\alpha_{l,Q}]^\top\in\mathbb C^B$. Stacking all delay blocks yields
$\bm{\alpha}\triangleq[\bm{\alpha}_0^\top,\ldots,
\bm{\alpha}_{L-1}^\top]^\top\in\mathbb C^{LB}$.
The average channel power is normalized as $\sum_{l=0}^{L-1}\sum_{q=-Q}^{Q}
\mathbb E[|\alpha_{l,q}|^2]=1$.

In this on-grid DD domain, the channel is often sparse: only a small number of DD coefficients are active. 
Classical compressed sensing (CS) based channel estimators treat $\bm{\alpha}$ as an unstructured sparse vector. They typically require either
\begin{itemize}[leftmargin=0.5cm]
    \item a priori knowledge of the sparsity level (e.g., the number of active paths), or
    \item a fixed sparsity budget or threshold that remains constant across channel realizations.
\end{itemize} 

However, in realistic doubly dispersive channels, the number of active delay taps and active Doppler components randomly varies from one realization to another. Consequently, a fixed sparsity budget inevitably leads to underestimation or overestimation. Moreover, conventional methods ignore the inherent structure of the DD support, which can be exploited to improve estimation reliability.

\subsection{Structured Sparsity in the DD Domain}\label{sec:IIB}
A general doubly sparse DD support can be factorized as $I_{l,q}=I_l I_q^{(l)}$, where $I_l$ indicates whether the $\ell$-th delay tap is active, and $I_q^{(l)}$ indicates which Doppler bins are active within that delay tap.
In many physically relevant scenarios, for instance, when the relative velocity between transmitter and receiver dominates the Doppler effect and the scattering environment exhibits a similar Doppler spread across different delays, the active Doppler bins are common to all active delay taps \cite{3gpp}.

Concretely, there exists a common indicator $S_q$ such that $I_q^{(l)}=S_q$ for all $l\in\{0,\ldots,L-1\}$, and therefore $I_{l,q}=I_l S_q$. 
The support variables are modeled as independent Bernoulli random variables: $I_l\sim\mathrm{Bernoulli}(p_d)$ and $S_q\sim\mathrm{Bernoulli}(p_D)$, with mutual independence across $\ell$ and across $q$. Moreover, the active complex gains are assumed independent of the support variables and satisfy $g_{l,q}\sim\mathcal{CN}(0,\sigma_\alpha^2)$.

Define the active delay set $\mathcal{D}^{\star}\triangleq\{\,l\in\{0,...,L-1\}: I_l=1\,\}$ and the active Doppler set $\mathcal{Q}^{\star}\triangleq\{\,q\in\{-Q,...,Q\}: S_q=1\,\}$. Then the true DD support has the Cartesian product form
\begin{equation}
\mathcal{S}^{\star}=\mathcal{D}^{\star}\times\mathcal{Q}^{\star}.
\label{eq:S_star}
\end{equation}
The cardinalities are $K_d \triangleq |\mathcal{D}^{\star}| = \sum_{l=0}^{L-1} I_l$ and $K_D \triangleq |\mathcal{Q}^{\star}| = \sum_{q=-Q}^{Q} S_q$.
Their means are $s_d \triangleq \mathbb{E}[K_d]=p_dL$ and $s_D \triangleq \mathbb{E}[K_D]=p_DB$, but the actual $K_d$ and $K_D$ vary per realization.

\begin{remark}[Unknown sparsity]
The estimator does not know $p_d$, $p_D$ nor the realization-specific $K_d$, $K_D$.
It only knows that the support obeys the product structure \eqref{eq:S_star}. The goal is to adaptively infer both the support dimensions and the coefficients from the received data.
\end{remark}

\subsection{AFDM Observation Model}
We now describe how the DD coefficient vector $\bm{\alpha}$ relates to the received AFDM pilot signals. Let $\bm{x}\!\in\!\mathbb{C}^{N}$ be the transmitted AFDM pilot. For a DD element at $(l,q)$, the corresponding received AFDM-domain sample at index $k$ is
\begin{equation}
\phi_{k,l,q}
=
\exp\!\left(
j2\pi
\left(
c_1 l^2 - \frac{ml}{N} + c_2(m^2-k^2)
\right)
\right)x_m,
\label{eq:atom_entry}
\end{equation}
where $m=(k-q+2Nc_1l)_N$ and $(\cdot)_N$ denotes the modulo-$N$ operation; $c_1$, $c_2$ are AFDM chirp parameters \cite{bemani2023affine}.

Let $\mathcal{P}\subseteq\{0,\ldots,N-1\}$ be the set of pilot observation indices (including guard samples), and let $M\triangleq |\mathcal{P}|$. Collecting the received DAFT-domain samples over $\mathcal{P}$ yields the observation vector $\bm{y}=[\,y_k\,]_{k\in\mathcal{P}}\in\mathbb{C}^{M}$. For each pair $(l,q)$, define the sensing atom $\bm{a}_{l,q}=[\,\phi_{k,l,q}\,]_{k\in\mathcal{P}}\in\mathbb{C}^{M}$. The AFDM sensing matrix is then constructed as
$\mathbf{M}_p
= [\bm{a}_{0,-Q}, \bm{a}_{0,-Q+1}, \cdots, \bm{a}_{0,Q},
\bm{a}_{1,-Q}, \cdots, \bm{a}_{L-1,Q}]
\in\mathbb{C}^{M\times LB}$.

Accordingly, the AFDM pilot observation model can be written compactly as
\begin{equation}
\bm{y} = \mathbf{M}_p \bm{\alpha} + \bm{w},
\label{eq:y_model}
\end{equation}
where $\bm{w}\sim\mathcal{CN}(\bm{0},\sigma_w^2\mathbf{I}_M)$. The mapping from a DD pair $(l,q)$ to the index in $\bm{\alpha}$ is given by $\iota(l,q) \triangleq lB + (q+Q+1), l\in\{0,\ldots,L-1\}, q\in\{-Q,\ldots,Q\}$. In other words, the coefficient associated with the DD-grid point $(l,q)$ is stored at the $\iota(l,q)$-th entry of $\bm{\alpha}$.

Overall, the channel estimation problem reduces to: Given $\bm{y}$ and $\mathbf{M}_p$, recover the sparse vector $\bm{\alpha}$ whose support obeys the Cartesian product structure \eqref{eq:S_star}, without knowing the sparsity levels $K_d$ and $K_D$ in advance.

\section{A Sparsity-Agnostic Channel Estimator}
\label{sec:algorithm}

In this section, we propose a new channel estimation method that does not require the knowledge of how many delay taps or Doppler bins are active. Instead, it directly estimates the channel from the received pilot signal by exploiting the fact that all active delays share the same Doppler support. The method works in two conceptual steps: (i) it generates candidate supports for different possible sparsity levels, and (ii) it selects the best candidate using BIC. The final channel estimate is obtained by a least-squares fit on the selected support.

\subsection{Problem Reformulation as Structured Model Selection}
Recall from \eqref{eq:S_star} that the true DD support has a Cartesian product structure.

For any candidate support $\mathcal{S}\subseteq \{0,\ldots,L-1\}\times\{-Q,\ldots,Q\}$, we denote by $\iota(\mathcal{S})=\{\iota(l,q):(l,q)\in\mathcal{S}\}$ the corresponding set of column indices. We now define a family of structured supports parameterized by the candidate cardinalities $d$ and $r$:
$\mathfrak{F}(d,r)
=
\Bigl\{
\mathcal{S}=\mathcal{D}\times\mathcal{Q}:
\mathcal{D}\subseteq\{0,\ldots,L-1\},\ |\mathcal{D}|=d,\ 
\mathcal{Q}\subseteq\{-Q,\ldots,Q\},\ |\mathcal{Q}|=r
\Bigr\}$.
The true support belongs to $\mathfrak{F}(K_d,K_D)$. The size of this structured family is
\begin{equation}
|\mathfrak{F}(d,r)|=\binom{L}{d}\binom{B}{r},
\label{eq:family_size_alg_new}
\end{equation}
which is substantially smaller than the number of arbitrary unstructured supports of size $dr$, i.e., $\binom{LB}{dr}$. The shared-Doppler prior thus provides a significant combinatorial reduction.

For any candidate support $\mathcal{S}\in\mathfrak{F}(d,r)$, let $\mathbf{M}_{\mathcal{S}}$ denote the submatrix of $\mathbf{M}_p$ obtained by selecting the columns indexed by $\iota(\mathcal{S})$. The least-squares estimate restricted to $\mathcal{S}$ is
\begin{equation}
\widehat{\bm{\alpha}}_{\mathcal{S}}
=
\argmin_{\supp(\bm{\alpha})\subseteq \iota(\mathcal{S})}
\|\bm{y}-\mathbf{M}_p\bm{\alpha}\|_2^2.
\label{eq:ls_support_restricted_alg_new}
\end{equation}
Let the corresponding residual sum of squares be
\begin{equation}
\mathrm{RSS}(\mathcal{S})
=
\|\bm{y}-\mathbf{M}_{\mathcal{S}}\widehat{\bm{\alpha}}_{\mathcal{S}}\|_2^2.
\label{eq:rss_alg_new}
\end{equation}
Under the Gaussian observation model, profiling out the noise variance yields the likelihood term
$M\log\!\left(\frac{\mathrm{RSS}(\mathcal{S})}{M}\right)$,
where $M=|\mathcal{P}|$ is the number of observations. Since a support in $\mathfrak{F}(d,r)$ contains $dr$ complex coefficients, its real-valued model dimension is $2dr$. We therefore define the BIC-type score
\begin{equation}
\mathcal{J}(\mathcal{S};d,r)
=
M\log\!\left(\frac{\mathrm{RSS}(\mathcal{S})}{M}\right)
+
2dr\log M.
\label{eq:bic_alg_new}
\end{equation}

The ideal structured model-selection problem is then
\begin{equation}
(\widehat d,\widehat r,\widehat{\mathcal{S}})
=
\arg\min_{\substack{0\le d\le L,\;0\le r\le B\\
\mathcal{S}\in\mathfrak{F}(d,r)}}
\mathcal{J}(\mathcal{S};d,r).
\label{eq:ideal_model_selection_alg_new}
\end{equation}
with final estimate $\widehat{\bm{\alpha}}
=\widehat{\bm{\alpha}}_{\widehat{\mathcal S}}$.


Problem \eqref{eq:ideal_model_selection_alg_new} is well-motivated but computationally prohibitive if we were to enumerate all supports in $\mathfrak{F}(d,r)$ exactly. We next develop a structured candidate-generation rule that is consistent with the shared-Doppler geometry and avoids combinatorial search within each $(d,r)$ hypothesis.

\subsection{Candidate Support Generation}
To address the challenge, our key idea is to compute a rough initial estimate of the DD coefficients that is robust to noise, and then use it to identify the most energetic delay and Doppler indices.

We first compute the regularized least-squares proxy
\begin{equation}
\bm{z}
=
\left(
\mathbf{M}_p^H\mathbf{M}_p+\lambda_{\rm reg}\mathbf{I}_{LB}
\right)^{-1}\mathbf{M}_p^H\bm{y},
\qquad \lambda_{\rm reg}>0,
\label{eq:proxy_alg_new}
\end{equation}
which is equivalently the unique minimizer of
\begin{equation}
\bm{z}
=
\argmin_{\bm{u}\in\mathbb{C}^{LB}}
\|\bm{y}-\mathbf{M}_p\bm{u}\|_2^2
+
\lambda_{\rm reg}\|\bm{u}\|_2^2.
\label{eq:ridge_proxy_alg_new}
\end{equation}
Here, $\lambda$ is a small regularization constant (e.g., $10^{-10}$). This is a standard ridge (Tikhonov) estimate \cite{Hoerl}. It is not meant to be the final channel, but it provides a smooth energy distribution across the DD grid, reducing noise spikiness.

Reshaping $\bm{z}$ blockwise yields the matrix
\begin{equation}
\mathbf{Z}
= [
\bm{z}_0^{\top},
\bm{z}_1^{\top},
\cdots,
\bm{z}_{L-1}^{\top}
]^\top
\in\mathbb{C}^{L\times B},
\label{eq:Z_alg_new}
\end{equation}
where the $\ell$-th row corresponds to the $\ell$-th delay block and the $q+Q+1$-th column corresponds to Doppler index $q\in\{-Q,\ldots,Q\}$.

For a candidate number of active Doppler bins $r$, we compute the total energy per Doppler bin summed over all delays
\begin{equation}
C(q)
\triangleq
\sum_{l=0}^{L-1}|Z_{l,q+Q+1}|^2,
\quad q\in\{-Q,\ldots,Q\}.
\label{eq:global_doppler_score_alg_new}
\end{equation}
We then pick the $r$ Doppler bins with the largest $C(q)$ as the candidate Doppler support $\widehat{\mathcal{Q}}(r)$, where
\begin{equation}
\widehat{\mathcal{Q}}(r)
=
\Top_r\{C(q)\}_{q=-Q}^{Q},
\label{eq:q_hat_alg_new}
\end{equation}
where $\Top_r\{\cdot\}$ returns the indices of the $r$ largest elements. This exploits the shared-Doppler property: if a Doppler bin is active, it should contribute significantly across all delay taps.

Given $\widehat{\mathcal{Q}}(r)$, we compute for each delay $\ell$ the energy restricted to those Doppler bins:
\begin{equation}
D_{\widehat{\mathcal{Q}}(r)}(l)
\triangleq
\sum_{q\in\widehat{\mathcal{Q}}(r)}
|Z_{l,q+Q+1}|^2,
\quad l\in\{0,\ldots,L-1\}.
\label{eq:conditional_delay_score_alg_new}
\end{equation}

Then, we choose the $d$ delays with the largest $D(l)$, yielding
\begin{equation}
\widehat{\mathcal{D}}(d,r)
=
\Top_d\{D_{\widehat{\mathcal{Q}}(r)}(l)\}_{l=0}^{L-1}.
\label{eq:d_hat_alg_new}
\end{equation}
The resulting candidate support for the pair $(d,r)$ is
\begin{equation}
\widehat{\mathcal{S}}(d,r)
=
\widehat{\mathcal{D}}(d,r)\times \widehat{\mathcal{Q}}(r).
\label{eq:s_hat_dr_alg_new}
\end{equation}

This construction is computationally cheap (only sorting and aggregation) and guarantees that every candidate support satisfies the Cartesian product structure.

\subsection{BIC-Based Model Order Selection}
For each pair $(d,r)$, we evaluate the BIC-type score on the candidate support \eqref{eq:s_hat_dr_alg_new}:
\begin{equation}
\widehat{\mathcal{J}}(d,r)
=
M\log\!\left(
\frac{\mathrm{RSS}(\widehat{\mathcal{S}}(d,r))}{M}
\right)
+
2dr\log M.
\label{eq:bic_practical_alg_new}
\end{equation}
The final estimated support dimensions are obtained by the exhaustive grid search
\begin{equation}
(\widehat d,\widehat r)
=
\argmin_{0\le d\le L,\ 0\le r\le B}
\widehat{\mathcal{J}}(d,r),
\label{eq:dr_hat_alg_new}
\end{equation}
and the final support estimate is
\begin{equation}
\widehat{\mathcal{S}}
=
\widehat{\mathcal{S}}(\widehat d,\widehat r).
\label{eq:s_hat_final_alg_new}
\end{equation}
The final channel estimate is then defined by
\begin{equation}
\widehat{\bm{\alpha}}
=
\argmin_{\supp(\bm{\alpha})\subseteq \iota(\widehat{\mathcal{S}})}
\|\bm{y}-\mathbf{M}_p\bm{\alpha}\|_2^2.
\label{eq:final_ls_alg_new}
\end{equation}

\begin{remark}
The proposed algorithm is one-stage in the sense that support-dimension selection and channel estimation are integrated into a single structured model-selection procedure. In particular, the method does not first assume $(d,r)$ and then apply a separate sparse-recovery routine; instead, it jointly searches over the admissible dimension grid and evaluates each candidate by a common BIC-type criterion.
\end{remark}

\subsection{Discussion}

The complete procedure is summarized in Algorithm~\ref{alg:proposed_structured_search_new}. The complexity is dominated by the ridge inversion (which can be precomputed if the pilot pattern is fixed) and the least-squares fits for each $(d,r)$. Since the grid $(L+1)(B+1)$ is modest, the algorithm can run efficiently.

The performance of Algorithm~\ref{alg:proposed_structured_search_new} is formalized in Proposition \ref{prop:oracle_equiv_alg_new}, wherein we prove the relation between the our channel estimator and the oracle least-squares benchmark (which knows the true support).

\begin{prop}
\label{prop:oracle_equiv_alg_new}
If the selected support coincides with the true support, i.e., $\widehat{\mathcal{S}}=\mathcal{S}^{\star}$, the proposed estimator coincides with the oracle least-squares estimator. In particular,
\begin{equation}
\widehat{\bm{\alpha}}
=
\argmin_{\supp(\bm{\alpha})\subseteq \iota(\mathcal{S}^{\star})}
\|\bm{y}-\mathbf{M}_p\bm{\alpha}\|_2^2.
\label{eq:oracle_equiv_statement_alg_new}
\end{equation}
Moreover, in the noiseless case, if $\mathbf{M}_{\mathcal{S}^{\star}}$ has full column rank, then $\widehat{\bm{\alpha}}=\bm{\alpha}$.
\end{prop}

\begin{proof}
(sketch) If $\widehat{\mathcal{S}}=\mathcal{S}^{\star}$, then \eqref{eq:final_ls_alg_new} is precisely the least-squares problem posed on the true support. Hence the estimator coincides with the oracle least-squares solution, which proves \eqref{eq:oracle_equiv_statement_alg_new}. In the noiseless case, if $\mathbf{M}_{\mathcal{S}^{\star}}$ has full column rank, then the least-squares solution on the true support is unique and must equal the true coefficient vector, which yields $\widehat{\bm{\alpha}}=\bm{\alpha}$.
\end{proof}

\begin{algorithm}[t]
\caption{Sparsity-agnostic channel estimation}
\label{alg:proposed_structured_search_new}
\begin{algorithmic}[1]
\STATE \textbf{Input:} $\bm{y}$, $\mathbf{M}_p$, $\lambda_{\rm reg}$.
\STATE Compute the regularized proxy $\bm{z}$ via \eqref{eq:proxy_alg_new}.
\STATE Reshape $\bm{z}$ into $\mathbf{Z}\in\mathbb{C}^{L\times B}$ via \eqref{eq:Z_alg_new}.
\STATE Initialize $\mathcal{J}_{\min}\leftarrow +\infty$.
\FOR{$r=0$ \TO $B$}
    \STATE Compute the candidate shared Doppler support $\widehat{\mathcal{Q}}(r)$ using \eqref{eq:q_hat_alg_new}.
    \FOR{$d=0$ \TO $L$}
        \STATE Compute the candidate delay support $\widehat{\mathcal{D}}(d,r)$ using \eqref{eq:d_hat_alg_new}.
        \STATE Form the candidate support $\widehat{\mathcal{S}}(d,r)=\widehat{\mathcal{D}}(d,r)\times\widehat{\mathcal{Q}}(r)$.
        \STATE Evaluate the BIC-type score $\widehat{\mathcal{J}}(d,r)$ via \eqref{eq:bic_practical_alg_new}.
        \IF{$\widehat{\mathcal{J}}(d,r)<\mathcal{J}_{\min}$}
            \STATE $\mathcal{J}_{\min}\leftarrow \widehat{\mathcal{J}}(d,r)$.
            \STATE Store $(\widehat d,\widehat r)\leftarrow(d,r)$ and $\widehat{\mathcal{S}}\leftarrow \widehat{\mathcal{S}}(d,r)$.
        \ENDIF
    \ENDFOR
\ENDFOR
\STATE \textbf{Output:} final channel estimate $\widehat{\bm{\alpha}}$.
\end{algorithmic}
\end{algorithm}

Our channel estimator addresses the fundamental issue that the true sparsity levels $K_d$ and $K_D$ are unknown and vary per channel realization. Fixed-budget approaches (e.g., always assuming a constant $d$ and $r$) will fail when the actual support dimensions differ. Our method, by contrast, selects the support dimensions directly from the observed data.



\section{Simulation Results}
\label{sec:simulation}
This section evaluates the proposed sparsity-agnostic channel estimator under the AFDM-based DD observation model. The primary goal is to verify whether the proposed method can accurately recover the unknown structured support and achieve near-oracle channel estimation performance without assuming the support cardinalities a priori.

\subsection{Simulation Setup}

We consider an AFDM frame with $N=4096$, $L=30$, $Q=7$, $B=2Q+1=15$,
and adopt the canonical AFDM pilot configuration with $
N_p=8$, $P_{\mathrm{afdm}}=2$, unless otherwise stated. The DAFT parameters are chosen as
$c_1=-\frac{P_{\mathrm{afdm}}}{2N}=-\frac{1}{4096}$, $
c_2=\frac{1}{20N}$.
Under this setting, the observation size is $|\mathcal{P}|=584$, and the corresponding AFDM pilot overhead equals $685$.

The channel follows the Bernoulli model in Section~\ref{sec:system_model} with delay activity probability $p_d=0.2$ and common Doppler activity probability $p_D=0.2$.
Hence, the mean support sizes are $s_d^{\rm mean}=p_dL=6$, $s_D^{\rm mean}=p_DB=3$
The noise variance follows the full-frame SNR convention $\sigma^2=\frac{N_p}{N}\cdot 10^{-\mathrm{SNR}/10}$, and the ridge regularization parameter is fixed at $\lambda_{\rm reg}=10^{-10}$ for all SNR points.
All curves are averaged over $5000$ independent Monte Carlo realizations.

\begin{itemize}[leftmargin=0.5cm]
\item \textit{Proposed method}: the sparsity-agnostic estimator with BIC selection.
\item \textit{Shared-mean baseline}: fixes support dimensions to the statistical means $(d,r)=(6,3)$, then constructs the support via the same structured rule. This tests the effect of a typical fixed budget.
\item \textit{Shared-tolerant baseline}: enlarges the fixed dimensions to $(d,r)=(12,6)$ to examine whether simply increasing the budget suffices.
\item \textit{Sparse Bayesian learning (SBL)}: an adaptive sparsity-promoting baseline that does not require a prescribed sparsity order. It uses the same sensing matrix and pilot overhead.
\item \textit{Oracle-LS}: least-squares on the true support, serving as a lower performance bound.
\end{itemize}

\subsection{The Normalized Mean-Square Error (NMSE) Performance}

\begin{figure}[t]
    \centering
    \includegraphics[width=0.8\columnwidth]{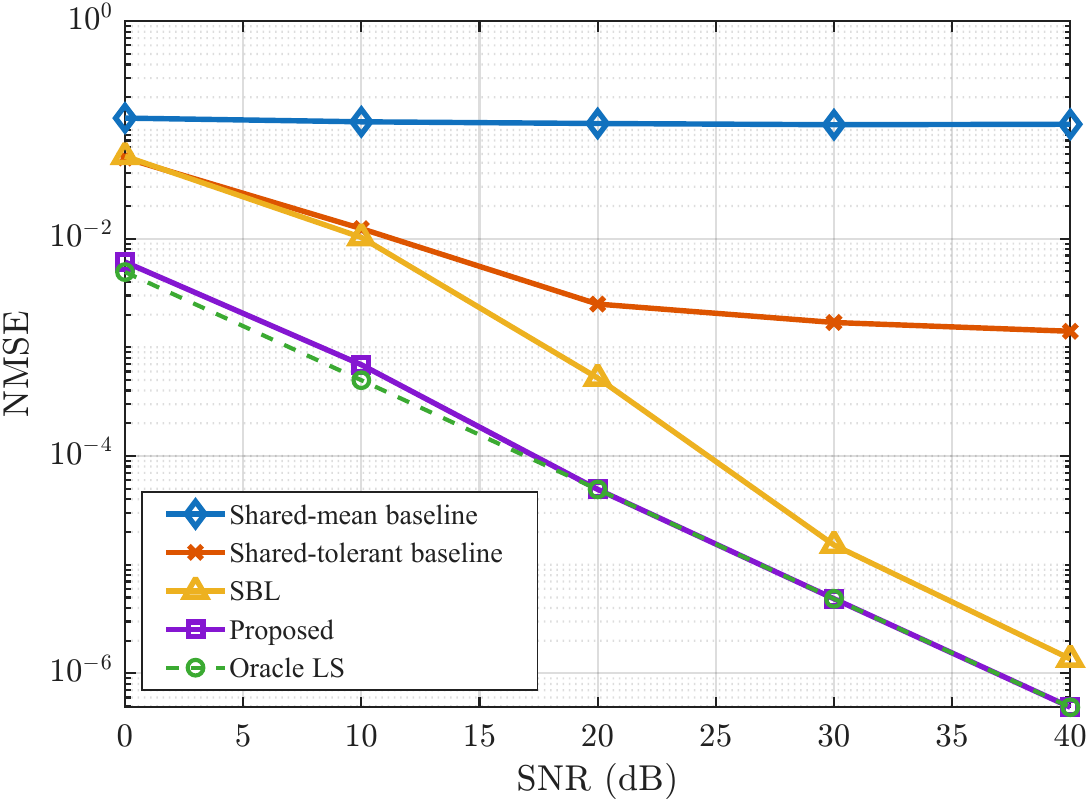}
    \caption{Average NMSE versus SNR for the considered channel estimation schemes.}
    \label{fig:nmse_main}
\end{figure}

Figure~\ref{fig:nmse_main} shows the average NMSE as a function of SNR for the considered channel estimation schemes. Unlike the fixed shared-count baselines, SBL does not require a prescribed sparsity order and instead promotes sparsity through automatic relevance determination. 

Several important observations can be made from Figure~\ref{fig:nmse_main}. 
\begin{enumerate}[leftmargin=0.5cm]
    \item First, the proposed method closely follows Oracle-LS over the entire SNR range. This indicates that the proposed structured support search can identify the correct DD support with high probability, so that the final least-squares reconstruction becomes essentially oracle-equivalent. In particular, the proposed method reaches the $10^{-4}$ NMSE regime at moderate SNR, confirming that the unknown support dimensions can be inferred directly from the received observations without sacrificing estimation accuracy.
    \item Second, the shared-mean baseline performs poorly across the whole SNR range and exhibits a pronounced error floor. This behavior shows that the dominant error source is not the noise level itself, but the mismatch caused by fixing the support dimensions to their statistical means. Since the true Bernoulli realization fluctuates around these means, a fixed shared-count assumption frequently underestimates the actual support and therefore removes informative coefficients.
    \item Third, the shared-tolerant baseline improves upon the shared-mean baseline, which confirms that enlarging the fixed support budget can partially reduce the risk of support underestimation. However, its performance remains significantly worse than that of the proposed method and still exhibits a clear high-SNR floor. This implies that merely increasing the support budget cannot resolve the underlying model-selection issue. Once the budget is overly enlarged, redundant support entries are retained, which increases the effective fitting dimension and introduces persistent false-alarm errors.
    \item Fourth, SBL achieves substantially better NMSE performance than the two fixed-budget baselines, especially at medium and high SNRs. This is expected because SBL does not rely on a manually prescribed support size and can adaptively suppress irrelevant coefficients through its hierarchical sparse prior. Nevertheless, SBL remains consistently worse than the proposed method and Oracle-LS. The performance gap indicates that automatic sparsity promotion alone is insufficient to fully exploit the adopted DS-LTV channel structure. In contrast, the proposed method explicitly searches over the structured DD support family and selects the realization-wise support dimensions through the BIC criterion, leading to more accurate support identification and channel reconstruction.
\end{enumerate}

Figure~\ref{fig:nmse_main} supports the central claim of the paper: under realization-wise support fluctuations, neither fixed mean-sized nor enlarged fixed budgets are adequate. SBL offers adaptivity but misses the structural prior, whereas the proposed method achieves near-oracle performance by explicitly leveraging the shared-Doppler geometry and data-driven order selection.

\subsection{Exact Support Recovery Probability}

\begin{figure}[t]
    \centering
    \subfloat[Proposed method: exact support recovery probability.]{
        \includegraphics[width=0.4\columnwidth]{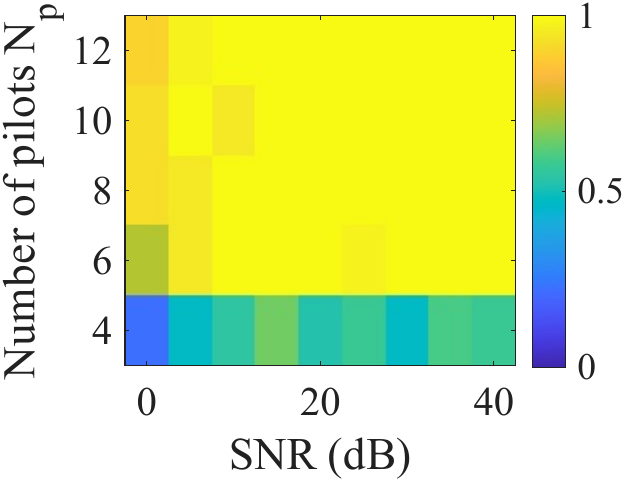}
        \label{fig:support_heatmaps_prop}
    }
    \hfill
    \subfloat[Shared-mean baseline: exact support recovery probability.]{
        \includegraphics[width=0.4\columnwidth]{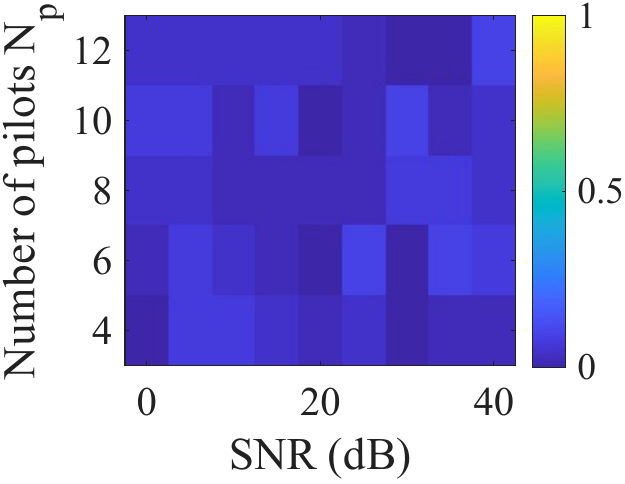}
        \label{fig:support_heatmaps_mean}
    }
    \caption{Exact support recovery probability over the $(N_p,\mathrm{SNR})$ plane.}
    \label{fig:support_heatmaps}
\end{figure}
To further understand the behavior observed in Figure~\ref{fig:nmse_main}, we next examine the exact support recovery probability. Figure~\ref{fig:support_heatmaps} presents two heatmaps over the two-dimensional grid formed by the SNR and the number of pilots $N_p$. Figure~\ref{fig:support_heatmaps}(a) corresponds to the proposed method, while Figure~\ref{fig:support_heatmaps}(b) corresponds to the shared-mean baseline.

From Figure~\ref{fig:support_heatmaps}(a), we observe that the proposed method exhibits a clear joint dependence on the pilot budget and the SNR. When the number of pilots is small, the support recovery probability is limited by the reduced number of observations, especially at low SNR. As $N_p$ increases, the support recovery probability rises rapidly, and once a moderate pilot budget is available, high support recovery probability is achieved in a broad SNR regime. This demonstrates that the proposed method effectively exploits the shared-Doppler structure and can reliably recover the exact support once the observation budget becomes sufficient.

In contrast, Figure~\ref{fig:support_heatmaps}(b) shows that the shared-mean baseline maintains a very low exact support recovery probability throughout almost the entire $(N_p,\mathrm{SNR})$ plane. Even when both SNR and pilot number increase, the recovery probability remains poor. This observation is particularly important because it indicates that the failure of the fixed shared-mean baseline is not primarily due to noise or observation deficiency. Instead, it is caused by the structural mismatch induced by the fixed support budget. In other words, increasing SNR or adding more pilots cannot fundamentally correct an incorrect model order assumption.


\subsection{Single-Realization Support Visualization}
\begin{figure}[t]
    \centering
    \subfloat[True.]{
        \includegraphics[width=0.22\columnwidth]{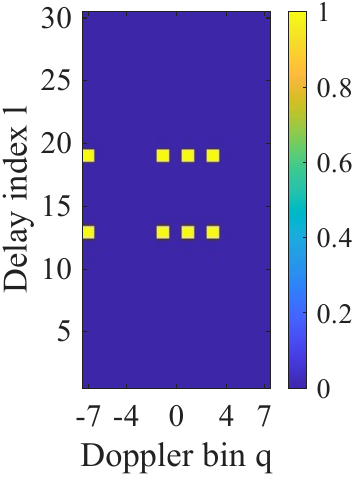}
        \label{fig:support_maps_true}
    }
    \hfill
    \subfloat[Proposed.]{
        \includegraphics[width=0.22\columnwidth]{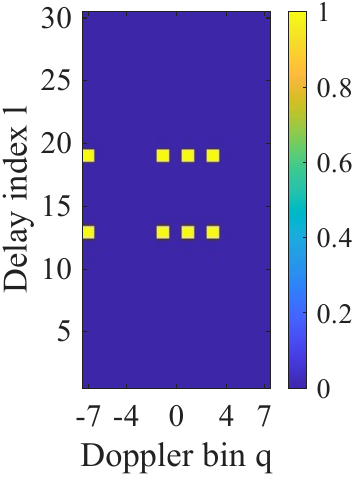}
        \label{fig:support_maps_prop}
    }
    \hfill
    \subfloat[Mean.]{
        \includegraphics[width=0.22\columnwidth]{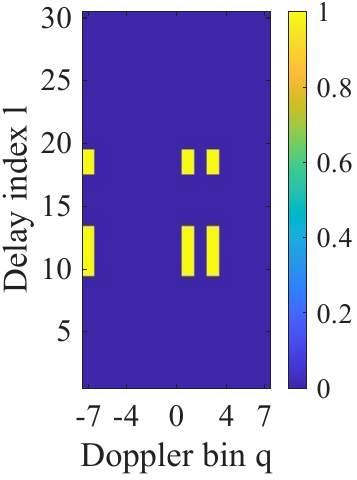}
        \label{fig:support_maps_mean}
    }
    \hfill
    \subfloat[Tolerant.]{
        \includegraphics[width=0.22\columnwidth]{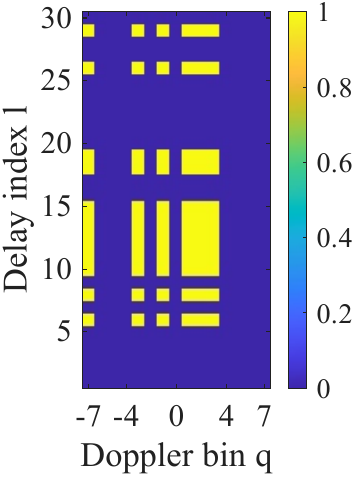}
        \label{fig:support_maps_tol}
    }
    \caption{Comparison of the true and recovered DD supports for one representative realization at \(\mathrm{SNR}=20\) dB.}
    \label{fig:support_maps}
\end{figure}
Figure~\ref{fig:support_maps} provides a qualitative visualization of the recovered supports for one representative realization at $\mathrm{SNR}=20$ dB. In this figure, the horizontal axis corresponds to the Doppler bin index $q\in\{-Q,\ldots,Q\}$, the vertical axis corresponds to the delay index $l\in\{0,\ldots,L-1\}$, and the highlighted entries indicate active coefficients.

The true support in Figure~\ref{fig:support_maps}(a) clearly exhibits the shared-Doppler structure: the active delay taps share a common Doppler support set. Figure~\ref{fig:support_maps}(b) shows that the proposed method reconstructs this structure almost exactly. Both the active delay set and the common Doppler support are correctly recovered, which is consistent with the near-oracle NMSE observed in Figure~\ref{fig:nmse_main}.

By contrast, Figure~\ref{fig:support_maps}(c) shows that the shared-mean baseline tends to underestimate the support. Although its assumed support family is consistent with the shared-Doppler geometry, the fixed mean-sized budget is too restrictive for many Bernoulli realizations, which leads to missing active coefficients. On the other hand, Figure~\ref{fig:support_maps}(d) shows that the shared-tolerant baseline overestimates the support by retaining an enlarged shared support set. This reduces the risk of misses, but introduces additional false alarms and an unnecessarily large fitting dimension.

Figure~\ref{fig:support_maps} visually confirms the interpretation of Figure~\ref{fig:nmse_main}: the shared-mean baseline suffers from underestimation, the shared-tolerant baseline suffers from overestimation, and the proposed method achieves accurate support identification by adaptively selecting the support dimensions from the data.


\section{Conclusion}
\label{sec:conclusion}

This paper developed a sparsity-agnostic DD-domain channel estimator that is practical for real-world systems, where sparsity varies unpredictably across time and propagation environments. The main takeaway is that for sparse channels with regular support, the right question is not ``how sparse?'' but ``what is the structure?'', and crucially, that structure can be learned from the data.
\bibliography{ref}

@article{bajwa2010compressed,
  title={Compressed channel sensing: A new approach to estimating sparse multipath channels},
  author={Bajwa, Waheed U and Haupt, Jarvis and Sayeed, Akbar M and Nowak, Robert},
  journal={Proceedings of the IEEE},
  volume={98},
  number={6},
  pages={1058--1076},
  year={2010},
}

@article{molisch2004geometry,
  title={Geometry-based directional model for mobile radio channels—Principles and implementation},
  author = {Molisch, Andreas F. and Kuchar, Alexander and Laurila, Juha and Hugl, Klaus and Schmalenberger, Ralph},
  journal = {European Transactions on Telecommunications},
  volume={14},
  number={4},
  pages={351--359},
  year={2003},
}

@article{shao2021federated,
  title={Federated edge learning with misaligned over-the-air computation},
  author={Shao, Yulin and G{\"u}nd{\"u}z, Deniz and Liew, Soung Chang},
  journal={IEEE Trans. Wire. Comm.},
  volume={21},
  number={6},
  pages={3951--3964},
  year={2021},
  publisher={IEEE}
}

@article{taubock2010compressive,
   author={Taubock, Georg and Hlawatsch, Franz and Eiwen, Daniel and Rauhut, Holger},
  journal={IEEE Journal of Selected Topics in Signal Processing}, 
  title={Compressive Estimation of Doubly Selective Channels in Multicarrier Systems: Leakage Effects and Sparsity-Enhancing Processing}, 
  year={2010},
  volume={4},
  number={2},
  pages={255-271},
 }

@article{berger2010application,
   title={Application of compressive sensing to sparse channel estimation},
  author={Berger, Christian R and Wang, Zhaohui and Huang, Jianzhong and Zhou, Shengli},
  journal={IEEE Communications Magazine},
  volume={48},
  number={11},
  pages={164--174},
  year={2010},
  publisher={IEEE}
}

@book{eldar2012compressed,
   title={Compressed sensing: theory and applications},
  author={Eldar, Yonina C and Kutyniok, Gitta},
  year={2012},
  publisher={Cambridge university press}
}

@article{bemani2023affine,
  title={Affine frequency division multiplexing for next generation wireless communications},
  author={Bemani, Ali and Ksairi, Nassar and Kountouris, Marios},
  journal={IEEE Trans. Wire. Comm.},
  volume={22},
  number={11},
  pages={8214--8229},
  year={2023},
  publisher={IEEE}
}

@article{zhang2025polarization,
  title={Polarization-Aware Movable Antenna},
  author={Zhang, Runxin and Shao, Yulin and Eldar, Yonina C},
  journal={IEEE Trans. Wire. Comm.},
  volume={25},
  pages={7428--7442},
  year={2025},
  publisher={IEEE}
}

@article{cao2025agile,
  title={Agile Affine Frequency Division Multiplexing},
  author={Cao, Yewen and Shao, Yulin},
  journal={arXiv preprint arXiv:2512.14424},
  year={2025}
}

@article{benzine2026models,
  title={Models, Methods, and Waveforms for Estimation and Prediction of Sparse Time--Varying Channels},
  author={Benzine, Wissal and Bemani, Ali and Ksairi, Nassar and Slock, Dirk},
  journal={IEEE Trans. Wire. Comm.},
  volume={25},
  pages={9623--9638},
  year={2025},
  publisher={IEEE}
}

@article{tropp2007signal,
  title={Signal recovery from random measurements via orthogonal matching pursuit},
  author={Tropp, Joel A and Gilbert, Anna C},
  journal={IEEE Trans. Info. Theory},
  volume={53},
  number={12},
  pages={4655--4666},
  year={2007},
  publisher={IEEE}
}

@article{duarte2011structured,
    author={Duarte, Marco F. and Eldar, Yonina C.},
  journal={IEEE Transactions on Signal Processing}, 
  title={Structured Compressed Sensing: From Theory to Applications}, 
  year={2011},
  volume={59},
  number={9},
  pages={4053-4085}
}

@article{schwarz1978estimating,
   ISSN = {00905364, 21688966},
 author = {Gideon Schwarz},
 journal = {The Annals of Statistics},
 number = {2},
 pages = {461--464},
 publisher = {Institute of Mathematical Statistics},
 title = {Estimating the Dimension of a Model},
 urldate = {2026-04-26},
 volume = {6},
 year = {1978}
 }

@techreport{3gpp,
  author      = "{3GPP}",
  title       = "{Study on Channel Model for Frequencies from 0.5 to 100 GHz}",
  institution = "{3GPP}",
  type        = "{Technical Report}",
  number      = "{TR 38.901}",
  month       = apr,
  year        = {2026}
}

@article{Hoerl,
author = {Arthur E. Hoerl and Robert W. Kennard},
title = {Ridge Regression: Biased Estimation for Nonorthogonal Problems},
journal = {Technometrics},
year = {1970},
publisher = {Taylor \& Francis}
}

@article{cao2026stepped,
  title={Stepped Frequency Division Multiplexing: A Jump-Free Continuous-Time {AFDM} Waveform},
  author={Cao, Yewen and Shao, Yulin},
  journal={arXiv:2605.11657},
  year={2026}
}

@article{shao2026embodied,
  title={Embodied Communication: Sensing-Induced Reliability Fields and Capacity Bounds},
  author={Shao, Yulin},
  journal={arXiv:2605.08284},
  year={2026}
}
\end{document}